\documentclass[12pt]{revtex4-1}
\usepackage{amsfonts,amsmath,amssymb,xcolor}
\usepackage{amsthm,amscd,ulem,bm}
\usepackage{graphicx}

\newcommand{\bib}[2]{\frac{\partial {#1}}{\partial {#2}}}
\def\w{{\wedge}} 

\def\div{\mbox{div}}

\theoremstyle{definition}
\newtheorem{pr}{Proposition}[section]
\newtheorem{defn}{Definition}[section]
\newtheorem{theo}{Theorem}[section]
\newtheorem{rem}{Remark}
\newtheorem{cor}{Corollary}[section]

\begin{document}
\title{
Super Finsler Connection of Superparticle \\ 
on Two Dimensional Curved Spacetime
}
\author{Takayoshi Ootsuka}
\email{ootsuka@cosmos.phys.ocha.ac.jp}
\affiliation{Department of Physics, Ochanomizu University, 2-1-1 Otsuka Bunkyo-ku, Tokyo, Japan}
\affiliation{NPO Gakujutsu-Kenkyu Network, Tokyo, Japan}
\author{Muneyuki Ishida}
\email{ishida@phys.meisei-u.ac.jp}
\affiliation{Department of Physics, Meisei University,
 2-1-1 Hodokubo, Hino, Tokyo 191-8506, Japan}
\author{Erico Tanaka}
\email{erico@sci.kagoshima-u.ac.jp}
\affiliation{Department of Mathematics and Computer Science, 
Kagoshima University, 1-21-35 K\={o}rimoto Kagoshima, Kagoshima, Japan}
\author{Ryoko Yahagi}
\email{yahagi@rs.tus.ac.jp}
 \affiliation{Department of Physics, Faculty of Science, 
Tokyo University of Science, 1-3 Kagurazaka, Shinjuku-ku, Tokyo, Japan}
\date{\today}

\begin{abstract}

We analyze the Casalbuoni-Brink-Schwarz superparticle model on a 2-dimensional curved spacetime as a super Finsler metric defined on a (2,2)-dimensional supermanifold. We propose a nonlinear Finsler connection which preserves this Finsler metric and calculates it explicitly. The equations of motion of the superparticle are reconstructed in the form of auto-parallel equations expressed by the super nonlinear connection.

\end{abstract}

\maketitle

\section{Introduction}

Einstein constructed the theory of gravity by considering the 
geodesic equations of a point particle.
Following his line of thought,
a natural way to construct a supergravity model is by the super
Riemannian formulation proposed by Arnowitt and Nath \cite{AN}.
They extended the standard Riemannian manifold to a supermanifold
that contains anticommuting Majorana spinors as coordinate functions.
They defined a connection, a curvature and field equations 
on this supermanifold through almost the same procedure 
as in the Einstein's gravity.
However, in spite of high expectations, the solutions of the field equations did not include 
a superspace with global supersymmetry \cite{Freu}.
Nevertheless, we still think that it is an ideal path to supergravity and 
believe that some modifications to the connection can salvage this method.

The simplest superparticle model was first given by Casalbuoni \cite{Cas1},
motivated by his study on the classical limit of fermion systems, which
was analyzed by Brink and Schwarz later on \cite{BS}.
The relation between the dynamics of this superparticle 
and the supergravity constraint equations 
via twistorial  interpretation was
suggested by Witten \cite{Witt}. 
Though the relation between superparticle and supergravity 
seems quite natural, 
this is the only literature that clearly states it.

In this paper, we consider the Casalbuoni-Brink-Schwarz 
2-dimensional superparticle Lagrangian as a super Finsler metric on a supermanifold 
(for the literature on supermanifold, we refer \cite{Rog, DeW} ).
We extend the nonlinear connection method on Finsler manifold invented by Kozma and Ootsuka \cite{KO} to a super Finsler manifold.
Despite the fact that an explicit calculation of Finsler connection is in general difficult, their technique makes the calculation much easier.
It is also applicable to a degenerate Finsler metric, which is a required property for our super Finsler metric.
In the major literature on Finsler geometry \cite{Mats, Mir, BCS, Ant, BM},
its connection defines parallel transports on the tangent bundle
(line element space).
Our parallel transports stay on the manifold (point space).
This standpoint is called 
point-Finsler approach \cite{IT,KT,KO}, and it is well suited for physical applications.
We do not need the linear part of the standard Finsler connection.
As for superconnection, our connection has a different definition
discussed by Bejancu \cite{Bej} 
and Vacaru and Vicol \cite{VV},
but the same one as difined by DeWitt \cite{DeW}.
We further extend the latter definition to a nonlinear connection so that it can be applied to degenerate super metrics.
Since this resembles Einstein’s theory of general relativity,
and taking the fact that the Casalbuoni-Brink-Schwarz model 
is a particle model with internal degrees of freedom (pseudoparticle)
into account, 
we are certain that it leads a theory of gravity for a
matter with internal degrees of freedom.
We also believe it corresponds to a supergravity without the Rarita-Schwinger field.

In section 2, we give a quick review of the spinor structure
and an analysis of the Casalbuoni-Brink-Schwarz model 
in terms of a super Finsler manifold.
Section 3 is devoted to a definition of a nonlinear Finsler connection on a supermanifold.
In section 4, we express the equations of motion of the superparticle as
auto-parallel equations.

\section{Casalbuoni-Brink-Schwarz model on curved super spacetime}

Casalbuoni-Brink-Schwarz superparticle model was originally defined on a flat spacetime.
We generalize this model to a curved spacetime as \cite{Witt}.
We consider 2-dimensional spacetime for simplicity
and present it as a super Finsler metric $L$ defined
on a (2,2)-dimensional super Finsler manifold $M^{(2,2)}$.
We take the even submanifold $M^{(2,0)}$ of $M^{(2,2)}$ as a 
Lorentzian manifold $(M^{(2,0)},g)$ and assume the Lorentzian metric
has signature $(+,-)$.
The dynamical variables of this model are $x^\mu \, (\mu=0,1)$, $\xi^A \, (A=1,2)$,
where $x^\mu$ represent spacetime even coordinates and $\xi^A$ are
Grassmann odd coordinates which are 
components of Majorana spinors $\xi = | _A \rangle \xi^A$.
The ket $| {}_A \rangle$ denotes spinor basis.
There exists an inner product $B_{AB}$ in the spinor space,
$B_{AB}=B_{BA}, B^2=1$ \cite{Chev, Suz},
which defines a cospinor of $|\xi \rangle$ 
as $\langle \xi | = \xi^A B_{AB} $ so that
$\langle \xi | \xi' \rangle = \xi^A B_{AB} \xi'^B$.
With cospinor basis $\langle ^A | $, the component of
the spinor can be
extracted as $\xi^A = \langle ^A | \xi \rangle$.
The matrix $B_{AB}$ satisfies $B \gamma^a B^{-1} = \, ^t \gamma^a$,
and $\gamma^a$ are the gamma matrices that admit the property
\begin{align}
 (\gamma^a)^A {}_C (\gamma^b)^C {}_B
 + (\gamma^b)^A {}_C (\gamma^a)^C {}_B
 = 2 \eta^{ab} \delta^A {}_B.
\end{align}
Here $a,b=0,1$ stand for the indices of flat spacetime whose metric is $\eta_{ab}$.
We use notations $(\gamma^a)_{AB}:=B_{AC}{(\gamma^a)^C}_B=(\gamma^a)_{BA}$.
For example, we can take gamma matrices and spinor metric
\begin{align}
 (\gamma^0)^A{}_B =
  \begin{pmatrix}
   0 & 1 \\
   1 & 0
  \end{pmatrix},
 \hspace{10mm}
 (\gamma^1)^A{}_B =
  \begin{pmatrix}
   0 & 1 \\
   -1 & 0
  \end{pmatrix},
\end{align}
and $B_{AB} = 
  \begin{pmatrix}
   0 & 1 \\
   1 & 0
  \end{pmatrix}$.

We start with the following Lagrangian
\begin{align}
 & L(x,dx,\xi,d\xi) =\sqrt{g_{\mu\nu}(x) \Pi^\mu \Pi^\nu}, \quad
 \Pi^\mu=dx^\mu + \langle \xi |\gamma^\mu(x) |d\xi\rangle \label{model} \\
 & \gamma^\mu(x):= \gamma^a e^\mu {}_a(x), \quad
 g=\eta_{ab} \theta^a \otimes \theta^b, \quad
 \theta^a={e^a}_\mu(x) dx^\mu, \label{mmodel}
\end{align}
where ${e^a}_\mu(x)$ are zweibeins.
We regard the Lagrangian $L$ as a super Finsler metric 
because it satisfies the properties of super Finsler metric 
described below.

We set $z^I:=(x^\mu, \xi^A)$, and capital Roman letters
starting from $I,J,\cdots$ stands for both spacetime and spinor indices.
$z^I$ and $dz^I$ satisfy the following commutation relations
\begin{align}
\begin{array}{l}
 z^I z^J = (-1)^{|I| |J|} z^J z^I, \\ 
 z^I dz^J = (-1)^{|I| |J|} dz^J z^I,  \\ 
 dz^I dz^J = (-1)^{|I| |J|} dz^J dz^I, 
\end{array}
\hspace{10mm}
 |I|=
 \begin{cases}
  0 \hspace{5mm} (I=\mu), \\
  1 \hspace{5mm} (I=A).
 \end{cases}
\end{align}
We also use the abbreviation $(z,dz)$ for $(z^I,dz^I)=(x^\mu,\xi^A,dx^\mu,d\xi^A)$.
The symbol $d$ is called a total derivative, and $dz^I$ plays a role of 
a coordinate function of the tangent space.
Namely, for a vector field $\displaystyle{v=\overleftarrow{\bib{}{z^I}} v^I}$
on the supermanifold $M^{(2,2)}$, it gives
\begin{align}
 dz^I(v)=v^I.
\end{align}
\begin{defn}
Suppose we have a well-defined differentiable function
$L:D(L) \subset TM^{(2,2)} \to \mathbb{R}$, where 
$D(L)$ is a subbundle of the tangent bundle $TM^{(2,2)}$.
$L$ is called {\it super Finsler metric} when it admits 
the homogeneity condition
\begin{align}
 \quad L(z, \lambda dz)=\lambda L(z,dz), \quad
 \lambda >0. \label{Fhom}
\end{align}
The set $(M^{(2,2)},L)$ is called a {\it supefr Finsler manifold}.
\end{defn}
We do not assume positivity: $L>0$ and regularity:
$\displaystyle \det \left( g_{IJ}(x,dx) \right) \neq 0$ for
$\displaystyle g_{IJ}(z,dz) := \frac{1}{2}\bib{^2 L^2}{dz^I \partial dz^J}$,
since these conditions are too strong for physical applications.
Note that the homogeneity condition implies
\begin{align}
 \bib{L}{dz^I} dz^I = L. \label{Fhom}
\end{align}

\begin{rem}
We call the super Finsler metric given by \eqref{model} and \eqref{mmodel}, 
{\it Casalbuoni-Brink-Schwarz metric}.
\end{rem}

Firstly, we show symmetries of the system
when the Lorentzian manifold $(M^{(2,0)},g)$ is flat,
$L=\sqrt{\eta_{ab} \Pi^a \Pi^b}, \quad 
\Pi^a=dx^a + \langle \xi |\gamma^a |d\xi\rangle$,
which has Poincar{\'e} symmetry and supersymmetry.
These symmetries are written in terms of vector fields on the supermanifold $M^{(2,2)}$.

\begin{defn}
The {\it Lie derivative of $L$ along a vector field $\displaystyle{v=\overleftarrow{\bib{}{z^I}} v^I}$
on the supermanifold }
is defined by 
\begin{align}
 {\cal L}_v L := \frac{\partial L}{\partial z^I} v^I 
 + \frac{\partial L}{\partial dz^I} dv^I.
\end{align} 
\end{defn}
\begin{defn}
A vector field $v$ is said to be a {\it Killing vector field}, when it satisfies
\begin{align}
 {\cal L}_v L = 0.
\end{align} 
\end{defn}
The Killing vector field which corresponds to the Lorentz transformation is
\begin{align}
 v = \overleftarrow{\bib{}{x^a}} \varepsilon^a {}_b x^b
 - \overleftarrow{\bib{}{\xi^A}} s^A {}_B \xi^B ,
 \hspace{10mm} 
 s^A {}_B:= -\frac18 [\gamma^a , \gamma^b]^A {}_B \varepsilon_{ab},
\end{align}
where $\varepsilon_{ab}$ and $s_{AB}$
are arbitrary anti-symmetric tensors,
$\varepsilon_{ab} := \eta_{ac}\varepsilon^c {}_b = - \varepsilon_{ba}$ and
$s_{AB}:=B_{AC} s^C {}_B=-s_{BA}$.
We can check that
\begin{align}
 {\cal L}_v L
 & = \frac{\Pi_a}{L} \left[ \bib{\Pi^a}{x^b} \varepsilon^b {}_c x^c
 + \bib{\Pi^a}{dx^b} \varepsilon^b {}_c dx^c
 - \bib{\Pi^a}{\xi^A} s^A {}_B \xi^B
 - \bib{\Pi^a}{d\xi^A} s^A {}_B d\xi^B
 \right] \\
 & = \frac{\Pi_a}{L} 
 \left[ \varepsilon^a {}_b dx^b
 + (\gamma^a)_{AC} d\xi^C s^A {}_B \xi^B
 - \xi^C (\gamma^a)_{CA} s^A {}_B d\xi^B
 \right] \\
 & = \frac{\Pi_a}{L} 
 \left[ \varepsilon^a {}_b dx^b
 + \varepsilon^a {}_b \xi^C (\gamma^b)_{CA} d\xi^A
 \right]
 = \frac{1}{L} \varepsilon^a {}_b \Pi_a \Pi^b
 =0,
\end{align}
where we used the identity $s^A{}_C (\gamma^a)^C {}_B 
- (\gamma^a)^A {}_C s^C {}_B = \varepsilon^a {}_b (\gamma^b)^A {}_B$.
For translation, we have
\begin{align}
 v = \overleftarrow{\bib{}{x^a}} \varepsilon^a,
\end{align}
with an arbitrary constant $\varepsilon^a$.
Supersymmetry transformation is described by
\begin{align}
 v = \overleftarrow{\bib{}{x^a}} \langle \xi | \gamma^a | \varepsilon \rangle
 + \overleftarrow{\bib{}{\xi^A}} \varepsilon^A
 = : Q_A \varepsilon^A,
 \hspace{10mm} 
 \frac12 \{Q_A,Q_B\} = \overleftarrow{\bib{}{x^a}} (\gamma^a)_{AB},
\end{align}
with an arbitrary Grassmann number $\varepsilon^A$.

Secondly, we derive the equations of motion of the model \eqref{model}
when $(M^{(2,0)},g)$ is not flat.
For convenience, we rewrite \eqref{model} using zweibeins.
\begin{align}
 L=\sqrt{\eta_{ab} \Pi^a \Pi^b}, \quad
 \Pi^a={\theta^a} + \langle \xi |\gamma^a |d\xi\rangle
 ={e^a}_\mu(x) \Pi^\mu. \label{zwei}
\end{align}
The action integral is given by the integration of the Finsler metric along 
an oriented curve $\bm{c}$ on $M^{(2,2)}$,
\begin{align}
 {\cal A}[\bm{c}]:=\int_{\bm{c}}L
 = \int_{t_0}^{t_1} c^* L
 = \int_{t_0}^{t_1}
 L\left(c^*z^I,c^*dz^I \right)
 = \int_{t_0}^{t_1}
 L\left(z^I (t),\frac{dz^I (t)}{dt}\right)dt,
\end{align}
where a map $c:t \in \mathbb{R} \mapsto c(t) \in M^{(2,2)}$ is 
a parametrization of the curve $\bm{c}$, and $c^* L$ represents
the pullback of $L$ by the map $c$.
The variation of the action is given by
\begin{align}
 \delta {\cal A}[\bm{c}] =\int_{\bm{c}} \delta L,
\end{align}
where
\begin{align}
 \delta L&=\frac{\eta_{ab}\Pi^a}{L}\delta \Pi^b
 =\frac{\Pi_a}{L} \left( \delta \left({e^a}_\mu dx^\mu\right)
 + \langle \delta \xi | \gamma^a | d\xi \rangle
 + \langle \xi | \gamma^a | d\delta \xi \rangle \right)
 \nonumber \\
 &=d\left\{\frac{\Pi_a}{L}{e^a}_\mu \delta x^\mu
         +\frac{\Pi_a}{L} \langle \xi | \gamma^a | \delta\xi \rangle
    \right\}
 +\frac{\Pi_a}{L}\partial_\mu {e^a}_\nu dx^\nu \delta x^\mu
 +\langle \delta \xi | \frac{\Pi_a \gamma^a}{L} |d\xi \rangle
 \nonumber \\
 & 
 -d\left(\frac{\Pi_\mu}{L}\right)\delta x^\mu 
 -d\left(\frac{\Pi_a}{L}\right) \langle \xi | \gamma^a | \delta \xi \rangle
 - \langle d\xi | \frac{\Pi_a \gamma^a}{L} | \delta \xi \rangle.
\end{align}
The Euler-Lagrange equations are extracted from the $\delta x^\mu$
part and $\delta \xi$ part:
\begin{align}
 0&=c^* \left[\frac{\Pi_a}{L}\partial_\mu e^a {}_\nu dx^\nu -d\left(\frac{\Pi_\mu}{L}\right) \right], \label{delx}
 \\
 0&=c^* \left[\frac{2 \Pi_a \gamma^a}{L} | d\xi \rangle
 +d\left(\frac{\Pi_a \gamma^a}{L}\right) | \xi \rangle\right]. \label{delxi}
\end{align}
These equations have terms including $d^2 z^I$, whose pullback by $c^*$ is given by
\begin{align}
 c^* d^2z^I = d \left( c^* dz^I \right) 
 = d \left( \frac{dz^I(t)}{dt} dt \right) 
 = \frac{d^2z^I(t)}{dt^2} (dt)^2.
\end{align}
Further on, we will omit the pullback symbol $c^*$ for notational simpilicity. 
We can rewrite \eqref{delx} as
\begin{align}
 0&=\frac{\Pi_b}{L}{e_a}^\mu \partial_\mu {e^b}_\nu dx^\nu
 -{e_a}^\mu d \left(\frac{\Pi_\mu}{L}\right)
 =\frac{\Pi_b}{L} e_a({e^b}_\nu) dx^\nu
 -d \left(\frac{{e_a}^\mu \Pi_\mu}{L}\right)
 +\frac{\Pi_\mu}{L}d{e_a}^\mu 
 \nonumber \\
 &=\frac{\Pi_b}{L}{\cal L}_{e_a}\theta^b-\frac{\Pi_\nu}{L}d{\cal L}_{e_a}x^\nu
 +\frac{\Pi_\mu}{L}d{e_a}^\mu-d\left(\frac{\Pi_a}{L}\right)
 =\frac{\Pi_b}{L}{\cal L}_{e_a}\theta^b-d\left(\frac{\Pi_a}{L}\right)
 \nonumber \\
 & =\frac{\Pi_b}{L} \iota_{e_a} \div \theta^b-d\left(\frac{\Pi_a}{L}\right),
\end{align}
where $\div$ is the exterior derivative of a 1-form,
$\displaystyle \div \theta^a := \frac12 \left( \partial_\mu e^a {}_\nu
- \partial_\nu e^a {}_\mu \right) dx^\mu \w dx^\nu$, 
which gives a 2-form, 
and $\iota_{e_a}$ the interior product.
With the above equation, \eqref{delxi} becomes
\begin{align}
 \frac{1}{L}\Pi_a \gamma^a | d\xi \rangle
 =-\frac1{2L}\Pi_b \iota_{e_a} \div \theta^b 
 \gamma^a | \xi \rangle. \label{mix}
\end{align}
Since the above equation is a first-order differential equation of $| \xi \rangle$,
it becomes a constraint on the supermanifold.
Multiplying it by $\displaystyle \frac{\Pi_c \gamma^c}{L}$, we have 
\begin{align}
 \frac{1}{L^2}\Pi_a \Pi_c \gamma^c \gamma^a |d\xi\rangle
 =\frac{1}{L^2}\Pi_a \Pi_c \left(\gamma^{ca}+ \eta^{ca}\right) |d\xi\rangle
 =|d\xi\rangle
 =-\frac{1}{2L^2} \Pi_b \Pi_c \gamma^c \gamma^a |\xi\rangle \iota_{e_a} \div \theta^b,
 \label{dxi0}
\end{align}
where $\displaystyle \gamma^{ca} 
:= \frac12 (\gamma^c \gamma^a-\gamma^a \gamma^c)$.
In the second equality, we used $\Pi_a \Pi_c \eta^{ca} = L^2$.
From the identity
\begin{align}
 \langle \xi |\gamma^n \gamma^c \gamma^a |\xi\rangle 
 =\langle \xi |\gamma^n \left(\gamma^{ca}+ \eta^{ca}\right)|\xi\rangle
 =\varepsilon^{ca}\langle \xi |\gamma^n \gamma^{01} |\xi\rangle
 = \varepsilon^{ca} \left(
 \eta^{n0} \langle \xi |\gamma^1 |\xi\rangle
 - \eta^{n1} \langle \xi |\gamma^0 |\xi\rangle
 \right),
\end{align}
and $\langle \xi |\gamma^a |\xi\rangle =0$,
we have
\begin{align}
 \langle \xi |\gamma^n |d\xi\rangle
 =-\frac{1}{2L^2} \Pi_b \Pi_c \langle \xi |\gamma^n \gamma^c \gamma^a |\xi\rangle 
 \iota_{e_a} \div \theta^b=0. \label{constraint}
\end{align}
Using together the definition of torsion
\begin{align}
 \displaystyle T^a=\div \theta^a + \omega^a {}_b \w \theta^b
 =\frac12 T^a {}_{bc} \theta^b \w \theta^c,
\end{align}
where $\omega^a {}_b$ is the spin connection,
equation \eqref{dxi0} becomes
\begin{align}
 |d\xi\rangle
 = -\frac{\Pi_b\Pi_c}{2L^2}
 \left(\iota_{e_a}T^b-{\omega^b}_{na}\theta^n+{\omega^b}_{an}\theta^n
 \right)\gamma^{c}\gamma^a|\xi\rangle
 =-\frac{\theta_b \theta_c}{2L^2}\left(\iota_{e_a}T^b-{\omega^b}_{na}\theta^n+{\omega^b}_{a}
 \right) \left(\gamma^{ca}+\eta^{ca}\right)|\xi \rangle.
\end{align}
The right hand side is calculated with an anti-symmetric tensor 
$\varepsilon^{ab} \, (\varepsilon^{01} = 1)$ as, 
\begin{align}
 &\theta_b \theta_c \iota_{e_a}T^b \left(\gamma^{ca}+\eta^{ca}\right)
 = - L^2 \theta_b T^b {}_{01} \gamma^{01}, \\
 &-\varepsilon^{ca}\theta^b \theta_c \theta^n \omega_{bna}
 +\varepsilon^{ca} \theta^b \theta_c \omega_{ba}
 =\{(\theta^0)^2-(\theta^1)^2 \}\omega_{01}
 =\Pi_a \Pi^a \omega_{01}=L^2 \omega_{01},
 \\
 &-\theta^b \theta^a \theta^n \omega_{bna}+\theta^b \theta^a \omega_{ba}
 =-\theta^b \theta^a \theta^n \omega_{bna}=0,
\end{align}
to result in 
\begin{align}
 |d\xi\rangle
 = \frac{1}{2} \left( \theta^b T_b {}_{01} - \omega_{01} \right)
 \gamma^{01}|\xi\rangle
 = - \frac{1}{2} \left( K_{01b} \theta^b + \omega_{01} \right)
 \gamma^{01}|\xi\rangle
 = - \frac{1}{4} \left( K_{abc} + \omega_{abc} \right)
 \gamma^{ab} \theta^c|\xi\rangle,
\end{align}
where we denoted the contorsion
\begin{align}
 K_{abc}:=\frac12 \left(T_{abc}+T_{bca}-T_{cab}\right).
\end{align}
Applying abbreviations such as
\begin{align}
 \hat{\omega}_c := \frac12 (K_{abc} + \omega_{abc}) \gamma^{ab},
 \hspace{10mm}
 \hat{\omega} := 
 (K_c + \omega_c) \theta^c
 = \frac12 (K_{abc} + \omega_{abc}) \gamma^{ab} \theta^c,
\end{align}
we obtain a simple expression
\begin{align}
 |d\xi\rangle
 + \frac{1}{2} \hat{\omega} |\xi\rangle = 0. \label{dxi}
\end{align}
With \eqref{constraint}, the equation \eqref{delx} becomes 
\begin{align}
 0&=\frac{\theta_a}{L}\partial_\mu e^a {}_\nu dx^\nu -d\left(\frac{g_{\mu\nu} dx^\nu}{L}\right)
 \nonumber \\
 & = \frac{1}{2L} \partial_\mu g_{\alpha\beta} dx^\alpha dx^\beta
 -d\left(\frac{g_{\mu\nu} dx^\nu}{L}\right).
 \label{eom}
\end{align}
After considering the pullback by $c^*$, this is exactly the same as the equation of motion of a relativistic free particle
on a Lorentzian manifold.
The Casalbuoni-Brink-Schwarz superparticle can be 
identifined as a relativistic particle with spin
obeying equation \eqref{dxi} as the internal degree of freedom.

\section{Nonlinear Finsler Connection on A Supermanifold}

In this section we will define a connection on a supermanifold 
which expresses naturally the geodesics of a superparticle. 
For this purpose we follow the definition given by Kozma and Ootsuka \cite{KO}. 
In their formulation, the Berwald connection is redefined as a nonlinear connection directly on point-Finsler space, and extended also to comprise the singular case. Such definition is advantageous for our purpose to consider the generalization to a supermanifold.  
We define a nonlinear generalization of the cotangent bundle $NT^* M^{(2,2)}$ as 
\begin{align}
 T^* M^{(2,2)} \subset NT^* M^{(2,2)} := \left\{
 a(z, dz) \, | \, a(z, \lambda dz) = \lambda a(z, dz), \lambda>0
\right\}.
\end{align}
A nonlinear 1-form $a \in NT^* M^{(2,2)}$ is a function of
$z^I$ and $dz^I$ and defines a map
$a: \Gamma(TM^{(2,2)}) \to C^\infty(M^{(2,2)})$,
$a(X) := a(z, dz(X)), \, X \in TM^{(2,2)}$ which 
satisfies a homogeneity condition $a(\lambda X) = \lambda a(X)$. 
It is not linear because 
$a(X+Y) \neq a(X) + a(Y)$.
\begin{defn} \label{connection}
Let $\Gamma(T^*M^{(2,2)})$ be a section of the cotangent bundle $T^*M^{(2,2)}$ on a supermanifold $M^{(2,2)}$
and $\nabla: \Gamma(T^*M^{(2,2)}) \to \Gamma(T^*M^{(2,2)} \otimes NT^*M^{(2,2)})$
a map such that satisfies
\begin{align}
 & \nabla dz^I = -N^I {}_J \otimes dz^J, \\
 & N^I {}_J (z,\lambda dz) = \lambda N^I {}_J (z,dz), \label{Nhom}\\
 & \bib{N^I {}_J}{dz^K} = (-1)^{|J| |K|} \bib{N^I {}_K}{dz^J}, \label{Ncond} \\
 & \bib{L}{z^I} = \bib{L}{dz^J} N^J {}_{I}. \label{Ndef}
\end{align}
Then, $N^I {}_J$ is called {\it a nonlinear super Finsler connection} on $M^{(2,2)}$.
\end{defn}
Unlike the linear connection, in general, 
$N^I {}_J$ is not linear in $dz^I$; namely
$N^I {}_J (z,dz) \neq N^I {}_{KJ} (z) dz^K$.
The condition \eqref{Nhom} means that the connection $N^I {}_J$
is degree 1 homogeneous.
For a nonlinear connection, the condition \eqref{Ncond} does not mean
the torsion is zero, while for the linear case, 
it becomes a torsion-free condition.
The last condition \eqref{Ndef} implies that the connection preserves 
the super Finsler metric:
$\displaystyle \nabla L:= \bib{L}{z^I} \nabla z^I + \bib{L}{dz^I} \nabla dz^I =0$.
We define the quantities $\displaystyle G^I := \frac12 N^I {}_J dz^J$ 
and call them {\it super Berwald functions}.
They are degree 2 homogeneity functions with respect to $dz^I$:
$G^I(z,\lambda dz)=\lambda^2 G^I(z,dz)$.
From this homogeneity condition, we have
\begin{align}
 \bib{G^I}{dz^J}
 = \frac12 \left(N^I {}_J + (-1)^{|J||K|} \bib{N^I {}_K}{dz^J} dz^K \right)
 = N^I {}_J. \label{GN}
\end{align}
\begin{rem}
The nonlinear connection defined above satisfies the linearity 
$\nabla(\rho_1+\rho_2)=\nabla \rho_1+\nabla \rho_2$
for sections $\rho_1$ and $\rho_2$ of $T^*M^{(2,2)}$,
which fails for the sections of $TM^{(2,2)}$.
Moreover, for physical problems, covariant quantities appear
more often than contravariant quantities do.
For these reasons, we proposed the definition \ref{connection}.
However, using such connection, we can also define 
the nonlinear connection 
for a vector field $\displaystyle{X=\overleftarrow{\bib{}{z^I}} X^I}$, by
\begin{align}
 \nabla X := \left( dX^I + N^I {}_{J}(z, dz(X) ) dz^J \right) 
 \otimes \overleftarrow{\bib{}{z^I}}.
\end{align}
Here $N^I {}_{J}(z^K, dz^K (X) )= N^I {}_{J}(z^K, X^K )$.
The connection above defines a map
$ \nabla: \Gamma(TM^{(2,2)}) \to 
\Gamma(T^*M^{(2,2)} \otimes TM^{(2,2)})$ with
$\nabla(\lambda X)=\lambda \nabla X, \, \lambda>0$ and
$\nabla(X+Y) \neq \nabla X + \nabla Y$.
\end{rem}

For the superparticle model, we have the following results on 
a nonlinear connection. 
\begin{theo}
Let $L$ be the Casalbuoni-Brink-Schwarz metric, 
then the super Berwald functions for $L$ and constraints are given by
\begin{align}
 & G^\mu
 = \frac12 \Gamma^\mu {}_{\alpha\beta} dx^\alpha dx^\beta
 + \frac{1}{L} \Pi^\mu \langle \mathcal{C} | d \xi \rangle
 - \frac{1}{2L^2} \Pi_a d e^a {}_\nu dx^\nu 
 \langle \xi | \gamma^\mu | d \xi \rangle
 - \frac12 g^{\mu\beta} \iota_{\partial_\beta} \div \theta^a
 \langle \xi | \gamma_a | d \xi \rangle
 + \langle \xi | \gamma^\mu | \lambda \rangle, \label{Gmu} \\
 & G^A
 = \frac{1}{2L^2} \Pi_a de^a {}_\nu dx^\nu d\xi^A 
 - \lambda^A, \label{GA} \\
 & \mathcal{C}_A := M_A - M_\mu (\gamma^\mu)_{AB} \xi^B=0,
\end{align}
where $\lambda^A$ are arbitrary functions of $(z^I,dz^I)$
which are second order homogeneous with respect to $dz^I$, and
\begin{align}
 M_\mu
 & := \frac{1}{2L} \left\{
 - \Pi_a \partial _\mu e^a {}_\nu dx^\nu
 + \left( \eta_{ab} - \frac{1}{L^2} \Pi_a \Pi_b \right) e^b {}_\mu de^a {}_\nu dx^\nu
 + \Pi_a de^a {}_\mu \right\},\\
 M_A
 & := \frac{1}{2L} \left\{
 2 \Pi_a \langle d\xi | \gamma^a | _A \rangle
 + \left( \eta_{ab} - \frac{1}{L^2} \Pi_a \Pi_b \right) de^a {}_\mu dx^\mu
 \langle \xi | \gamma^b | _A \rangle \right\}.
\end{align}
\end{theo}
\begin{proof}
Firstly, we multiply \eqref{Ndef} by $dz^I$ from the right and obtain
\begin{align}
 \bib{L}{z^I} dz^I
 = \bib{L}{dz^J} N^{J} {}_{I} dz^I
 = 2 \bib{L}{dz^I} G^I.
\end{align}
Considering the homogeneity condition \eqref{Fhom},
we find a particular solution for $G^I$:
\begin{align}
 G^I = \frac12 \left( \bib{L}{z^J} dz^J \right) \frac{dz^I}{L}.
\end{align}
Since we are considering (2,2)-dimensional supermanifold,
we need 4 independent vectors to span the general solution.
We choose vectors 
\begin{align}
 l^I {}_\mu
 : = \delta^I {}_\mu - \frac{\Pi_\mu}{L^2} dz^I, \label{sol1}
 \hspace{10mm}
 l^I {}_B : =
  \begin{pmatrix}
   l^\mu {}_B \\
   l^A {}_B
  \end{pmatrix}
 : =
  \begin{pmatrix}
   \langle \xi | \gamma^\mu | _B \rangle \\
   - \langle ^A | _B \rangle
  \end{pmatrix} 
\end{align}
for the basis.
It is easy to check that these vectors vanish when they are contracted 
with $\displaystyle \bib{L}{dz^I}$ from the left.
Thus, we can write the general solution as
\begin{align}
 G^I = \frac12 \left( \bib{L}{z^J} dz^J \right) \frac{dz^I}{L}
 + l^I {}_\mu \lambda^\mu + l^I {}_A \lambda^A, \label{Gdef}
\end{align}
where $\lambda^\mu, \lambda^A$ are arbitrary functions of $(z^I,dz^I)$,
and are second order homogeneous with respect to $dz^I$.
Since there are 5 non-independent vectors $(dz^I, l^I {}_\mu, l^I {}_B)$ in the solution, we can choose
one additional condition for the coefficients $\lambda^\mu$. 
For convenience, we set
\begin{align}
 \Pi_\mu \lambda^\mu = 0. \label{pilam}
\end{align}

For further calculation, we define 
\begin{align}
 L_{IJ}
 & := \bib{^2 L}{dz^I \partial dz^J}
 := L \overleftarrow{\bib{}{dz^I}} \overleftarrow{\bib{}{dz^J}}, \\
 L_{\mu\nu}
 & = \bib{^2 L}{dx^\mu \partial dx^\nu}
 = \frac{1}{L} \left( g_{\mu\nu} - \frac{1}{L^2} \Pi_\mu \Pi_\nu \right)
 = L_{\nu\mu}, \\
 L_{A\nu}
 & = \bib{^2 L}{d\xi^A \partial dx^\nu}
 = \langle \xi | \gamma^\alpha L_{\alpha\nu} | _A \rangle
 = L_{\nu A}, \\ 
 L_{AB}
 & = \bib{^2 L}{d\xi^A \partial d\xi^B}
 = \langle \xi | \gamma^\alpha | _A \rangle L_{\alpha\beta}
 \langle \xi | \gamma^\beta | _B \rangle
 = - L_{BA}.
\end{align}
From \eqref{GN}, \eqref{sol1}, and \eqref{Gdef} we have
\begin{align}
 N^I {}_J
 & = \bib{G^I}{dz^J} \notag \\
 & = \frac12 \left( \bib{L}{z^K} dz^K \right)
 \left( \frac{1}{L} \delta^I {}_J - \frac{dz^I}{L^2} \bib{L}{dz^J} \right)
 + \frac12 (-1)^{|I||J|} \left( \bib{L}{z^J} + \bib{^2L}{dz^J \partial z^K} dz^K\right)
 \frac{dz^I}{L} \notag \\
 & \hspace{5mm}
 + l^I {}_\mu \bib{\lambda^\mu}{dz^J}
 + \bib{l^I {}_\mu}{dz^J} \lambda^\mu
 + l^I {}_A \bib{\lambda^A}{dz^J}
 + (-1)^{|J|} \bib{l^I {}_A}{dz^J} \lambda^A.
\end{align}
The last term will vanish due to \eqref{sol1}.
When this is multiplied by $\displaystyle \bib{L}{dz^I}$, only few terms remain:
\begin{align}
 \bib{L}{dz^I} N^{I} {}_{J}
 = \frac12 \left( \bib{L}{z^J} + \bib{^2L}{dz^J \partial z^K} dz^K \right)
 - \lambda^\mu L_{\mu J}. \label{seki}
\end{align}
The relation \eqref{Ndef} says that the left hand side
of \eqref{seki} is equal to 
$\displaystyle \bib{L}{z^J}$, which leads to
\begin{align}
 \lambda^\mu L_{\mu J}
 = \frac12 \left( -\bib{L}{z^J} + \bib{^2L}{dz^J \partial z^K} dz^K \right)
 =: M_J \label{M}.
\end{align}
We separate the above equation into two pieces.
For $J=A$, we leave it as a constraint 
$\mathcal{C}_A = M_A - \lambda^\mu L_{\mu A}=0$.
For $J=\mu$, we rewrite it into a matrix equation
\begin{align}
  \begin{pmatrix}
   L_{\mu\nu} & \frac{\Pi_\mu}{L} \\
   \frac{\Pi_\nu}{L} & 0
  \end{pmatrix}
  \begin{pmatrix}
   \lambda^\nu \\
   0
  \end{pmatrix}
 =
  \begin{pmatrix}
   M_\mu \\
   0
  \end{pmatrix}.
\end{align}
The second row is the condition \eqref{pilam}.
This matrix has the inverse matrix
\begin{align}
  \begin{pmatrix}
   \tilde{L}^{\mu\nu} & \frac{\Pi^\mu}{L} \\
   \frac{\Pi^\nu}{L} & 0
  \end{pmatrix},
 \hspace{10mm}
 \tilde{L}^{\mu\nu} := L g^{\mu\nu} - \frac{\Pi^\mu \Pi^\nu}{L},
\end{align}
and we have
\begin{align}
  \begin{pmatrix}
   \lambda^\mu \\
   0
  \end{pmatrix}
 =
  \begin{pmatrix}
   \tilde{L}^{\mu\nu} M_\nu \\
   \frac{1}{L} \Pi^\nu M_\nu
  \end{pmatrix}. \label{lamb0}
\end{align}
With this $\lambda^\mu$, we obtain
\begin{align}
 G^I = \frac12 \left( \bib{L}{z^J} dz^J \right) \frac{dz^I}{L}
 + l^I {}_\mu \tilde{L}^{\mu\nu} M_\nu + l^I {}_A \lambda^A. \label{sol2}
\end{align}
The second row of \eqref{lamb0} automatically holds. 
This can be checked by considering the constraint
\begin{align}
 0 = \mathcal{C}_A
 = M_A - \lambda^\mu \langle \xi | \gamma^\alpha L_{\alpha\mu} | _A \rangle
 = M_A - M_\alpha \langle \xi | \gamma^\alpha | _A \rangle,
\end{align}
where \eqref{M} is used for the last equality.
Taking the contraction with $d\xi^A$, we get
\begin{align}
 0
 = \langle \mathcal{C} | d\xi \rangle
 = M_A d\xi^A - M_\alpha (\Pi^\alpha - dx^\alpha)
 = M_I dz^I - M_\alpha \Pi^\alpha
 = - M_\alpha \Pi^\alpha. \label{cdxi}
\end{align}
For the last equality, we used \eqref{Fhom} to obtain
\begin{align}
 M_I dz^I
 = \frac12 \left( -\bib{L}{z^I} dz^I + \bib{^2L}{dz^I \partial z^K} dz^K dz^I \right)
 = 0.
\end{align}
The explicit expressions of $M_I$ are calculated straightforward.

Using the relation between the Christoffel symbol and zweibeins, 
\begin{align}
 \Gamma_{\mu\alpha\beta} dx^\alpha dx^\beta
 & = \frac12 (\partial_\alpha g_{\mu\beta} + \partial_\beta g_{\alpha\mu}
 - \partial_\mu g_{\alpha\beta}) dx^\alpha dx^\beta \notag \\
 & = \eta_{ab} (e^b {}_\beta \partial_\alpha e^a {}_\mu
 + e^a {}_\mu \partial_\alpha e^b {}_\beta
 - e^a {}_\alpha \partial_\mu e^b {}_\beta) dx^\alpha dx^\beta,
\end{align}
we obtain
\begin{align}
 2 L M_\mu
 = \Gamma_{\mu\alpha\beta} dx^\alpha dx^\beta
 - \frac{1}{L^2} \Pi_a \Pi_b e^b {}_\mu de^a {}_\nu dx^\nu
 - \langle \xi | \gamma_a | d \xi \rangle \iota_{\partial_\mu} \div \theta^a.
\end{align}
With the above relation and
\begin{align}
 \bib{L}{z^J} dz^J
 = \frac1L \Pi_a de^a {}_\mu dx^\mu,
\end{align}
the even part of the super Berwald function $G^\mu$ becomes
\begin{align}
 G^\mu
 & = \frac12 \left( \bib{L}{z^J} dz^J \right) \frac{dx^\mu}{L}
 + l^\mu {}_\alpha \tilde{L}^{\alpha\beta} M_\beta + l^\mu {}_A \lambda^A \notag \\
 & = \frac{1}{2L^2} \Pi_a de^a {}_\nu dx^\nu dx^\mu
 + \left( L g^{\mu\beta} - \frac{1}{L} \Pi^\mu \Pi^\beta \right) M_\beta
 + \langle \xi | \gamma^\mu | \lambda \rangle \notag \\
 & = \frac{1}{2L^2} \Pi_a de^a {}_\nu dx^\nu dx^\mu
 + \frac12 \Gamma^\mu {}_{\alpha\beta} dx^\alpha dx^\beta
 - \frac{1}{2L^2} \Pi_a de^a {}_\nu dx^\nu \Pi^\mu
 - \frac12 g^{\mu\beta} \iota_{\partial_\beta} \div \theta^a
 \langle \xi | \gamma_a | d \xi \rangle \notag \\
 & \hspace{5mm} - \frac{1}{L} \Pi^\mu \Pi^\beta M_\beta
 + \langle \xi | \gamma^\mu | \lambda \rangle \notag \\
 & = \frac12 \Gamma^\mu {}_{\alpha\beta} dx^\alpha dx^\beta
 + \frac{1}{L} \Pi^\mu \langle \mathcal{C} | d \xi \rangle
 - \frac{1}{2L^2} \Pi_a d e^a {}_\nu dx^\nu 
 \langle \xi | \gamma^\mu | d \xi \rangle
 - \frac12 g^{\mu\beta} \iota_{\partial_\beta} \div \theta^a
 \langle \xi | \gamma_a | d \xi \rangle
 + \langle \xi | \gamma^\mu | \lambda \rangle.
\end{align}
In the last line, we used the identity \eqref{cdxi}.
For the odd part $G^A$, the relation
\begin{align}
 l^A {}_\alpha \tilde{L}^{\alpha\beta}
 = -\frac{1}{L^2} \Pi_\alpha \tilde{L}^{\alpha\beta} d\xi^A
 = 0
\end{align}
assures
\begin{align}
 G^A
 & = \frac12 \left( \bib{L}{z^J} dz^J \right) \frac{d\xi^A}{L}
 + l^A {}_\alpha \tilde{L}^{\alpha\beta} M_\beta + l^A {}_B \lambda^B \notag \\
 & = \frac{1}{2L^2} \Pi_a de^a {}_\nu dx^\nu d\xi^A
 - \lambda^A .
\end{align}
\end{proof}

Note that the super Berwald functions \eqref{Gmu} and \eqref{GA}
are nonlinear with respect to $dz^I$.
This result cannot arise if linear connections are assumed from the start
as in \cite{AN}. 
We think this is why they cannot construct the supergravity.
Only nonlinear connection is allowed for the 
Casalbuoni-Brink-Schwarz model.

Without odd variables, that is $\xi^A=0$, the connection $N^\mu {}_\nu$ becomes the usual Riemannian connection.
Therefore, our formulation is a natural extension.

\begin{pr}
The constraint $\mathcal{C}_A=0$ is equivalent to the equation \eqref{mix}
and eventually leads to 
\eqref{constraint}, $\langle \xi |\gamma^a |d\xi\rangle=0$.
\end{pr}
\begin{proof}
From the definition of $M_I$, we have
\begin{align}
 2 \mathcal{C}_A
 & = 2 M_A - 2 M_\mu (\gamma^\mu)_{AB} \xi^B \notag \\
 & = -\bib{L}{\xi^A} + \bib{L}{d\xi^A \partial x^\mu} dx^\mu
 + \bib{^2L}{d\xi^A \partial \xi^B} d\xi^B \notag \\
 & \hspace{4mm}
 - \left( -\bib{L}{x^\mu} + \bib{^2L}{dx^\mu \partial x^\nu} dx^\nu
 + \bib{^2L}{dx^\mu \partial \xi^C} d\xi^C \right) (\gamma^\mu)_{AB} \xi^B.
\label{2ca}
\end{align}
We put the results
\begin{align}
 \bib{L}{\xi^A}
 & = - \frac{\Pi_\mu}{L} (\gamma^\mu)_{AB} d\xi^B, \\
 \bib{^2L}{d\xi^A \partial x^\mu}
 & = \bib{}{x^\mu} \left( \frac{\Pi_\nu}{L} \right)
 (\gamma^\nu)_{AB} \xi^B
 + \frac{\Pi_\nu}{L} \partial_\mu e_a {}^\nu (\gamma^a)_{AB} \xi^B, \\
 \bib{^2L}{d\xi^A \partial \xi^B}
 & = \frac{\Pi_\mu}{L} (\gamma^\mu)_{AB}
 - \left( \frac{\Pi_\mu}{L} \right) \overleftarrow{\bib{}{\xi^B}} 
 (\gamma^\mu)_{AC} \xi^C, \\
 \bib{L}{x^\mu}
 & = \frac{\Pi_a}{L} \partial_\mu e^a {}_\nu dx^\nu, \\
 \bib{^2L}{dx^\mu \partial x^\nu}
 & = \bib{}{x^\nu} \left( \frac{\Pi_\mu}{L} \right), \\
 \bib{^2L}{dx^\mu \partial \xi^C}
 & = \left( \frac{\Pi_\mu}{L} \right) \overleftarrow{\bib{}{\xi^C}} ,
\end{align}
into equation \eqref{2ca}, and obtain
\begin{align}
 2 \mathcal{C}_A
 & = \frac{\Pi_\mu}{L} (\gamma^\mu)_{AB} d\xi^B
 + \bib{}{x^\mu} \left( \frac{\Pi_\nu}{L} \right) dx^\mu
 (\gamma^\nu)_{AB} \xi^B
 + \frac{\Pi_\nu}{L} \partial_\mu e_a {}^\nu dx^\mu (\gamma^a)_{AB} \xi^B
 \notag \\
 & \hspace{4mm}
 + \frac{\Pi_\mu}{L} (\gamma^\mu)_{AB} d\xi^B
 - \left( \frac{\Pi_\mu}{L} \right) \overleftarrow{\bib{}{\xi^B}} 
 (\gamma^\mu)_{AC} \xi^C d\xi^B \notag \\
 & \hspace{4mm}
 + \frac{\Pi_a}{L} \partial_\mu e^a {}_\nu dx^\nu (\gamma^\mu)_{AB} \xi^B
 - \bib{}{x^\nu} \left( \frac{\Pi_\mu}{L} \right)
 dx^\nu (\gamma^\mu)_{AB} \xi^B 
 - \left( \frac{\Pi_\mu}{L} \right) \overleftarrow{\bib{}{\xi^C}}
 d\xi^C (\gamma^\mu)_{AB} \xi^B \notag \\
 & = \frac{2\Pi_\mu}{L} (\gamma^\mu)_{AB} d\xi^B
 + \frac{\Pi_b}{L} (e^b {}_\nu \partial_\mu e_a {}^\nu
 + e_a {}^\nu \partial_\nu e^b {}_\mu )
 dx^\mu (\gamma^a)_{AB} \xi^B \notag \\
 & = \frac{2\Pi_a}{L} (\gamma^a)_{AB} d\xi^B
 + \frac{\Pi_b}{L} \iota_{e_a} \div \theta^b (\gamma^a)_{AB} \xi^B.
\end{align}
Thus, $\mathcal{C}_A=0$ means the equation \eqref{mix}.
\end{proof}

\section{Auto-parallel equations}

To rewrite the Euler-Lagrange equations into auto-parallel equations, 
\eqref{Ncond} is the key condition.
With these nonlinear super Finsler connections, we have the following result.
\begin{theo}
The Euler-Lagrange equations of the superparticle are expressed as 
the auto-parallel equations
\begin{align}
 0 & = c^* \biggl[ d^2 x^\mu
 + \Gamma^\mu {}_{\alpha\beta} dx^\alpha dx^\beta
 + \frac{2}{L} \Pi^\mu \langle \mathcal{C} | d \xi \rangle
 - \frac{1}{L^2} \Pi_a d e^a {}_\nu dx^\nu 
 \langle \xi | \gamma^\mu | d \xi \rangle
 - g^{\mu\beta} \iota_{\partial_\beta} \div \theta^a
 \langle \xi | \gamma_a | d \xi \rangle \notag \\
 & \hspace{4mm}
 - \frac{\lambda}{L} dx^\mu
 - \langle \xi | \gamma^\mu | \lambda \rangle \biggr],  \label{d2x}\\
 0 & = c^* \biggl[ d^2 \xi^A
 + \frac{1}{L^2} \Pi_a de^a {}_\nu dx^\nu d\xi^A
 - \frac{\lambda}{L} d\xi^A + \lambda^A \biggr], \label{d2xi}
\end{align}
with the constraint
\begin{align}
 c^* \left( \mathcal{C}_A \right) = 0.
\end{align}
Here, $\lambda$ and $\lambda^A$ are arbitrary functions of $(z^I,dz^I)$,
and are second order homogeneous with respect to $dz^I$.
\end{theo}
\begin{proof}
We start with the Euler-Lagrange equation. Making use of 
condition \eqref{Ncond}, we have
\begin{align}
 0
 & = \bib{L}{z^I} - d \left( \bib{L}{dz^I} \right) \notag \\
 & = \bib{L}{z^I} - \bib{^2 L}{dz^I \partial z^J} dz^J
 - \bib{^2 L}{dz^I \partial dz^J} d^2z^J \notag \\
 & = \bib{L}{z^I} - (-1)^{|I| |J|} \left( \bib{L}{z^J} \right) \overleftarrow{\bib{}{dz^I}} dz^J
 - L_{IJ} d^2z^J \notag \\
 & = \bib{L}{dz^J} N^J {}_I - (-1)^{|I| |J|}
 \left( \bib{L}{dz^K} N^K {}_J \right) \overleftarrow{\bib{}{dz^I}} dz^J
 - L_{IJ} d^2z^J \notag \\
 & = \bib{L}{dz^J} N^J {}_I - \left\{ (-1)^{|I| |J|} \bib{L}{dz^K} \bib{N^K {}_J}{dz^I}
 + (-1)^{|I| |K|} L_{KI} N^K {}_J \right\} dz^J
 - L_{IJ} d^2z^J \notag \\
 & = - L_{IJ} ( d^2z^J + 2 G^J ).
\end{align}
Since
\begin{align}
 L_{IJ} \frac{dz^J}{L} = 0,
 \hspace{10mm}
 L_{IJ} l^J {}_A = 0,
\end{align}
we can expand it as
\begin{align}
 d^2z^I + 2 G^I = \lambda \frac{dz^J}{L} + l^I {}_A \lambda^A, \label{zg}
\end{align}
with arbitrary functions $\lambda$ and $\lambda^A$.
Substituting \eqref{Gmu} and \eqref{GA} to \eqref{zg}, and redefining
the arbitrary function $\lambda^A$ 
using homogeneity conditions, we have the desired results.
\end{proof}
The arbitrary functions $\lambda$ and $\lambda^A$ have different origins:
$\lambda$ emerges from the reparametrization invariance and 
 is determined when the time parameter is fixed,
and $\lambda^A$, or $| \lambda \rangle$ in the bracket notation,
is related to the gauge symmetries (constraint $\mathcal{C}_A=0$).
For this superparticle model, it is determined by the consistency 
with the equation \eqref{dxi}.
\begin{cor}
Suppose the constraint $\mathcal{C}_A=0$ is satisfied.
Then we have
\begin{align}
 d^2 x^\mu
 & = - \Gamma^\mu {}_{\alpha\beta} dx^\alpha dx^\beta
 + \frac{\lambda}{L} dx^\mu \\
 d^2 \xi^A
 & = \frac{\lambda}{L} d\xi^A - \frac12 
 \left\{ d (\hat{\omega}_c)^A {}_B \theta^c \xi^B
 - (\hat{\omega}_c)^A {}_B \omega^c {}_{ab} \theta^a \theta^b \xi^B
 + (\hat{\omega})^A {}_B d\xi^B \right\}.
\end{align}
\end{cor}
\begin{proof}
When $\mathcal{C}_A=0$, the equations \eqref{d2x} and \eqref{d2xi} become
\begin{align}
 d^2 x^\mu
 & = - \Gamma^\mu {}_{\alpha\beta} dx^\alpha dx^\beta
 + \frac{\lambda}{L} dx^\mu
 + \langle \xi | \gamma^\mu | \lambda \rangle, \label{d2x2} \\
 d^2 \xi^A
 & = - \frac{1}{L^2} \eta_{ab} \theta^b de^a {}_\nu dx^\nu d\xi^A
 + \frac{\lambda}{L} d\xi^A - \lambda^A. \label{d2xi2}
\end{align}
To evaluate the parameter $ | \lambda \rangle$,
take the total derivative of \eqref{dxi},
\begin{align}
 |d^2 \xi \rangle
 = d \left\{ - \frac12 \hat{\omega}_c \theta^c | \xi \rangle \right\}
 = - \frac12 \left\{ d\hat{\omega}_c \theta^c | \xi \rangle
 + \hat{\omega}_c d \theta^c | \xi \rangle
 + \hat{\omega} | d \xi \rangle \right\}.
\end{align}
For the part $d \theta^c$, we have
\begin{align}
 d \theta^c
 & = d (e^c {}_\mu dx^\mu) \notag \\
 & = de^c {}_\mu dx^\mu + e^c {}_\mu d^2x^\mu \notag \\
 & = de^c {}_\mu dx^\mu
 - e^c {}_\mu \Gamma^\mu {}_{\alpha\beta} dx^\alpha dx^\beta
 + \frac{\lambda}{L} \theta^c
 + \langle \xi | \gamma^c | \lambda \rangle \notag \\
 & = - \omega^c {}_{ab} \theta^a \theta^b
 + \frac{\lambda}{L} \theta^c
 + \langle \xi | \gamma^c | \lambda \rangle, \label{dth}
\end{align}
where we substituted \eqref{d2x2} into $d^2 x^\mu$ in the second line.
Then we obtain
\begin{align}
 |d^2 \xi \rangle
 = - \frac12 \biggl\{ d\hat{\omega}_c \theta^c | \xi \rangle
 - \hat{\omega}_c \omega^c {}_{ab} \theta^a \theta^b | \xi \rangle
 + \frac{\lambda}{L} \hat{\omega}_c \theta^c | \xi \rangle 
 + \langle \xi | \gamma^c | \lambda \rangle \hat{\omega}_c | \xi \rangle 
 + \hat{\omega} | d \xi \rangle \biggr\}.
\end{align}
The third term becomes $\displaystyle \frac{\lambda}{L} | d\xi \rangle$
due to the equation \eqref{dxi}.
Comparing this and the equation \eqref{d2xi2}, we obtain
\begin{align}
 | \lambda \rangle
 = - \frac{1}{L^2} \eta_{ab} \theta^b de^a {}_\mu dx^\mu | d\xi \rangle
 + \frac12 d \hat{\omega}_c \theta^c | \xi \rangle
 - \frac12 \hat{\omega}_c \omega^c {}_{ab} \theta^a \theta^b |\xi\rangle
 + \frac12 \hat{\omega} | d\xi \rangle.
\end{align}
Put this $ | \lambda \rangle$ back into \eqref{d2x2} and \eqref{d2xi2}, and 
the result follows.
\end{proof}

\begin{rem}
With \eqref{dxi} and \eqref{dth}, the equation \eqref{d2xi2}
becomes 
\begin{align}
 d^2 \xi^A = -\frac12 d \hat{\omega}^A {}_B \xi^B
 + \frac14 \hat{\omega}^A {}_B \omega^B {}_C \xi^C,
\end{align}
and this is equivalent to 
\begin{align}
 D(D |\xi\rangle)=0,
 \hspace{10mm}
 D |\xi\rangle
 := |d\xi\rangle
 + \frac{1}{2} \hat{\omega} |\xi\rangle.
\end{align}
\end{rem}

By the terminology of constrained systems, we can say that
$\mathcal{C}_A=0$ is a second-class constraint,
since the Lagrange multiplier $\lambda^A$ is determined,
as suggested in \cite{Cas1, BS} for the flat case.

\section{Discussion}

In this paper, we have newly defined a nonlinear connection on a super Finsler manifold and calculate it in the case of the Casalbuoni-Brink-Schwarz metric.
We have expressed how the equations of motion
of the superparticle are rewritten in the form of the auto-parallel equations.
Our explicit calculation displays the nonlinear connection truly 
plays a critical role in this process,
though the last corollary indicates that the connection would become linear 
after exposed the constraint $\mathcal{C}_A=0$.
This setup is fundamentally different from the one in Arnowitt-Nath \cite{AN}
where only a linear connection is used.
Considering the fact that our procedure is similar to Einstein's approach 
to a relativistic particle in his theory of general relativity,
we are on the right track to construct a theory of supergravity
form a superparticle.
The Casalbuoni-Brink-Schwarz model leads
a theory of gravity for a matter with internal degrees of freedom.
To prove it, we are now working on the derivation of the induced connection,
Finsler curvature, and torsion on the constraints.
For supergravity, the system 
with an additional Rarita-Schwinger field is underway as well.
We also note that this method is applicable to any higher dimensional systems, which is remarkable because an explicit calculation of Finsler connection is difficult even in a 2-dimensional case.

\begin{acknowledgements}

We thank Prof. L. Kozma, Prof. M. Morikawa and Prof. A. Sugamoto for valuable discussions.

\end{acknowledgements}



\end{document}